\documentclass[conference]{IEEEtran}
\IEEEoverridecommandlockouts

\usepackage{amsmath,amssymb,amsfonts}
\usepackage{tabularx,booktabs}
\usepackage{mathtools}
\usepackage{algorithmic}
\usepackage{graphicx}
\usepackage{textcomp}
\usepackage{xcolor}
\usepackage{dsfont}
\usepackage{hyperref}
\usepackage{nccmath}
\def\BibTeX{{\rm B\kern-.05em{\sc i\kern-.025em b}\kern-.08em
    T\kern-.1667em\lower.7ex\hbox{E}\kern-.125emX}}
\newcolumntype{C}{>{\centering\arraybackslash}X} % centered version of "X" type
\newcolumntype{b}{>{\hsize=2.3\hsize}X}
\usepackage{amsthm}
\usepackage{bbm}
\usepackage{csquotes}
\usepackage{chngcntr}
\usepackage{apptools}
\usepackage{enumitem}
\usepackage{cite}

\theoremstyle{plain}
\newtheorem{theorem}{Theorem}
\newtheorem*{theorem*}{Theorem$\,\diamond$\!\!}
\newtheorem{lemma}{Lemma}
\newtheorem{corollary}{Corollary}
\newtheorem{proposition}{Proposition}

\theoremstyle{definition}
\newtheorem{definition}{Definition}

\theoremstyle{remark}

\newtheorem{remark}{Remark}
\newtheorem{example}{Example}

\newcommand{\A}{\mathcal{A}}

\newcommand{\X}{\mathcal{X}}
\newcommand{\Y}{\mathcal{Y}}
\newcommand{\Z}{\mathcal{Z}}
\newcommand{\M}{\mathcal{M}}
\newcommand{\Hy}{\mathcal{H}}

\newcommand{\E}{\mathbb{E}}

\newcommand{\Pm}{\mathcal{P}}

\makeatletter
\newcounter{labelcnt}
\renewcommand{\thelabelcnt}{(\alph{labelcnt})}
\newcommand{\setlabel}[1]{%
  \refstepcounter{labelcnt}\ltx@label{lbl:#1}%
  {\text{\upshape\thelabelcnt}}%
}

\makeatother

\providecommand{\MG}[1]{{\color{red} MG: #1}}

\begin{document}

\title{From Generalisation Error to Transportation-cost Inequalities and Back}

\author{
\IEEEauthorblockN{Amedeo Roberto Esposito, Michael Gastpar}
\IEEEauthorblockA{\textit{School of Computer and Communication Sciences} \\
       EPFL, Lausanne, Switzerland\\
\{amedeo.esposito, michael.gastpar\}@epfl.ch}
}

\maketitle
\begin{abstract} In this work, we connect the problem of bounding the expected generalisation error with transportation-cost inequalities. 
 Exposing the underlying pattern behind both approaches we are able to generalise them
 and go beyond Kullback-Leibler Divergences/Mutual Information and sub-Gaussian measures. In particular, we are able to provide a result showing the equivalence between two families of inequalities: one involving functionals and one involving measures. This result generalises the one proposed by Bobkov and G\"otze that connects transportation-cost inequalities with concentration of measure. Moreover, it allows us to recover all standard generalisation error bounds involving mutual information and to introduce new, more general bounds, that involve
arbitrary divergence measures.
\end{abstract}

\section{Introduction}
A recent and interesting line of research has explored the problem of bounding the generalization error of learning algorithms via information measures \cite{infoThGenAn,learningMI, wassersteinCalmon,generrWassDist,genErrMISGLD,tighteningMI,fullVersionGeneralization, ismi, conditionalMI}. The starting observation is that one can interpret a learning algorithm as a (potentially) randomised mapping that takes as input a data-set and provides as output a hypothesis (e.g., a classifier). The purpose is to retrieve a classifier that has good performance both on the training set and on an \textit{independent} test set. Intuitively, if the outcome of the learning algorithm \textbf{depends} too much on its input then its performance on new data-points will be poor. In this case, the algorithm is said to \textit{overfit} to the training set. One way of measuring said dependence is through information measures. To this day, virtually every information measure has been connected to the problem: Mutual Information \cite{infoThGenAn, learningMI, ismi, conditionalMI}, Total Variation \cite{wassersteinCalmon},  R\'enyi's $\alpha$-Divergences, Sibson's $\alpha-$Mutual Information, $f-$divergences and $f-$Mutual Information \cite{fullVersionGeneralization}, etc.
The connection has been drawn with respect to expected generalisation error \cite{infoThGenAn, explBiasMI, wassersteinCalmon, generrWassDist} and with respect to the probability of having a large generalisation error \cite{fullVersionGeneralization}.\\
We noticed a similarity between the bound connecting expected generalisation error to mutual information and transportation-cost inequalities. Transportation-cost inequalities connect Wasserstein distances and Kullback-Leibler divergences to the concentration of measure phenomenon.
In this work we will connect the dots and expose the pattern connecting these objects. Moreover, unearthing the underlying mechanism we will generalise both type of results (generalisation-error bounds and transportation-cost inequalities) and go beyond Kullback-Leibler Divergences and sub-Gaussian tails. 
%\subsection{Related Works}

\section{Background and definitions}
Following the De Finetti's notation, instead of denoting the expectation of a function $f$ with respect to a probability measure $\mu$ with $\mathbb{E}_\mu[f]$ we will denote it with $\mu(f)$. Spaces will be denoted with calligraphic letters $\X,\Y$, random variables with capital letters $X,Y$ and measures with Greek lower-case letters $\mu,\nu$.
Given a function $f:\X\to \Y$ we will denote with $f^\star$ its Legendre-Fenchel dual, defined as follows.
\begin{equation}
    f^\star(x^\star)= \sup_{x\in \X} \langle x, x^\star \rangle - f(x), \label{lfDuality}
\end{equation}
where $\langle x, x^\star \rangle$ denotes the natural pairing between a space $\X$ and its topological dual $\X^\star$, \textit{i.e.},  $\langle x, x^\star \rangle= x^\star(x).$ A different but related notion of duality is Young's duality. Given a function $f:\mathbb{R}\to\mathbb{R}$ we denote, through a slight abuse of notation, with $f^\star$ its Young's complementary function, defined as follows:
\begin{equation}
    f^\star(x^\star)= \sup_{x>0} \langle x,x^\star \rangle - f(x),\,\,\, x^\star>0. \label{youngsDuality}
\end{equation}
Notice that the supremum in \eqref{youngsDuality} is over a different set with respect to \eqref{lfDuality}. We will always specify which notion of duality we are using throughout the paper.
Given two measures $\mu,\nu$ and a convex functional $\varphi$ we denote with $D_\varphi(\nu\|\mu)=\mu\left(\varphi\left(\frac{d\nu}{d\mu}\right)\right)$ the $\varphi$-divergence, where $\frac{d\nu}{d\mu}$ is the Radon-Nikodym derivative. With $\varphi(x)=x\log x$ one recovers the KL-divergence $D(\nu\|\mu)$, with $\varphi(x)=|x-1|$ one recovers the Total Variation distance $TV(
\nu,
\mu)$ and with $\varphi(x)=\frac{|x|^\alpha}{\alpha}, \alpha>0$ one recovers the Hellinger integral $H_\alpha(\nu\|\mu)$.
\subsection{Optimal Transport Theory}\label{tcI} \iffalse
The origins of Optimal Transport date back to Monge's work in 1781. In his formulation, the problem represented the formalisation of a very intuitive and practical issue: redistribution (in the sense of transportation and reshaping) of some material (e.g., sand or soil) with the minimum effort/cost. The problem, as stated by Monge, does not always admit a solution\cite{ambrosioOT}. Almost 200 years later, the problem was resurrected and rescued from oblivion by Kantorovich, who proposed a simple but powerful relaxation of the problem \cite{kantorovich}. \fi Denote with $\mathcal{M}(\X)$ the set of all signed measures over $\X$ and with $\mathcal{P}(\X)$ the set of all probability measures over $\X$. Let $\mu,\nu \in \mathcal{P}(\X)$, consider the set $\Pi(\mu,\nu)$ of all the joint probability measures $\pi \in \mathcal{P}(\X\times\X)$ with marginals equal to $\mu$ and $\nu$, \textit{i.e.}, such that $\pi(\cdot\times\X)= \mu(\cdot)$ and $\pi(\X\times\cdot)=\nu(\cdot)$. The problem advanced by Kantorovich was the following: given $\mu$ and $\nu$ and a Borel function $d:\X\times\X\to [0,+\infty]$, can we find a joint measure $\pi$ that minimises $
    \pi(d(X,Y))?$ $\pi$ represents a \textbf{transport plan} between $\mu,\nu$. 
Under mild assumptions on $d$, optimal transport plans are guaranteed to exists for extremely general spaces $\X$ \cite[Theorem 4.1]{VillaniOToldandnew}. %Kantorovich's idea was to consider the elements $\pi$ in $\Pi(\mu,\nu)$ as \textbf{transport} \textit{plans} (a joint measure over $\X\times\X$), rather than the transport maps (a deterministic function) proposed by Monge. 
%Such transport plans always exist: given $\mu$ and $\nu$ one can always construct the joint measure $\mu\times\nu$. %Moroeover, any transport map induces a transport plan and a transport plan can be seen as induced by a transport map under mild conditions \cite[Prop 2.1]{ambrosioOT}. 
%Clearly, any transport map induces a transport plan.
%The riformulated problem was, given $\mu$ and $\nu$ and a Borel function $c:\X\times\X\to [0,+\infty]$, can we find a joint measure $\pi$ that minimises $
%    \pi(c(X,Y))?$
%Under mild assumptions on $c$, % (\textit{i.e.}, lower semi-continuity and a lower-bound required to ensure the well-definedness of $\mathbb{E}_\pi[c(X,Y)]$ in $\mathbb{R}\cup\{+\infty\}$), 
 %optimal transport plans are guaranteed to exists for extremely general spaces $\X$ \cite[Theorem 4.1]{VillaniOToldandnew}.
%\subsection{Wasserstein Distances}
If $d$ itself is a metric over $\X$, then $\inf_\pi\pi(d(X,Y))$ represents a distance over the space of probability measures. %: the larger is the quantity the more ``difficult'' it is to transform $\mu$ into $\nu$.
 %\\One can see $\inf_\pi\pi(c(X,Y))$ as some sort of distance between $\mu$ and $\nu$: the larger is the quantity the more ``difficult'' it is to transform $\mu$ into $\nu$.
 %However, unless $c$ itself is a metric on $\X$, $\inf_\pi\pi(c(X,Y))$ does not represent a distance over the space of probability measures, in general. %However, unless one makes further assumptions on $c$, this object does not satisfy the axioms of a distance.% However, if we restrict the function $c$ to be a metric on $\X$ we can retrieve a distance over the space of probability measures. 
 In particular, given a metric $d$, let us denote with $\Pm_p(\X)$ the set of probability measures $\mu$ on $\X$ such that $\mu(d(X,x_0)^p)^{1/p}<+\infty$ for some $x_0\in \X$.
 \begin{definition}[\!\!\!{{\cite[Def. 6.1]{VillaniOToldandnew}}}]
 Let $(\X,d)$ be a Polish space and $p\in[1,+\infty).$ Let $\mu,\nu \in \Pm_p(\X)$, the $p-$Wasserstein distance between $\mu$ and $\nu$ is defined as $
     W_p(\mu,\nu) = \inf_{\pi\in\Pi(\mu,\nu)} (\pi(d(X,Y)^p))^{1/p}.$
 \end{definition}
 Wasserstein distances satisfy interesting properties \cite[ Lemma 3.4.1]{concentrationMeasureII}. Moreover, when connected to KL-divergences, in what are known in the literature as ``Transportation-cost Inequalities'', they have interesting implications in the concentration of measure phenomenon. 
 \iffalse\begin{lemma}\cite{villaniOToldandnew}
 Let $(\X,d)$ be a Polish space we have that:
 \begin{itemize}
     \item if $p\geq 1$ $W_p$ is a metric on $\Pm_p(\X)$;
     \item if $1\leq p\leq q$ then $W_p(\mu,\nu)\leq W_q(\mu,\nu) \, \forall \mu,\nu\in\Pm_q(\X)$;
     \item $W_p$ metrizes weak convergence in $\Pm_p(\X).$
 \end{itemize}
 \end{lemma}
\subsection{Transportation-cost inequalities} \fi %\label{tcI} 
%A bridge connecting Optimal Transport theory and Information Theory can be found in the transportation-cost inequalities.
\begin{definition}[Transportation-Cost Inequality]\label{tpcIneq}
Let $(\X,d)$ be a Polish space and $\mu$ a probability measure on $\X$, we say that $\mu$ satisfies an $L^p$-transportation-cost inequality with constant $c$ (or $T_p(c)$ in short) if for every $\nu\ll\mu$
\begin{equation}
    W_p(\mu,\nu)\leq \sqrt{2cD(\nu\|\mu)}.\label{classicalTpc}
\end{equation}
\end{definition}

%Transportation cost inequalities are strongly connected to Gaussian-like concentration of measure.
When $p=1$, for instance, these inequalities are equivalent to concentration in the following sense:
\begin{theorem*}[\!\!\!{{\cite[Thm 3.1]{bobkov}}}]
Let $\mu\in\Pm_1(\X)$ be a Borel probability measure. There exists a $c$ such that for every $\lambda\in\mathbb{R}$
\begin{equation}
    \log\mu(\exp(\lambda f)) \leq \left(\frac{c \lambda^2}{2}\right) \label{tpcHypothesis}
\end{equation}
for every $1-$Lipschitz function $f$, if and only if $\mu$ satisfies a $T_1(c)$ inequality, \textit{i.e.}, for every $\nu\ll\mu$
\begin{equation}
    W_1(\mu,\nu)\leq \sqrt{2cD(\nu\|\mu)} \label{tpcThesis}.
\end{equation}
%The measure $\mu$ satisfies a $T_1(c)$ inequality, \textit{i.e.,}if and only if, for every $1-$Lipschitz function $f:\X\to\mathbb{R}$ with $\mu(f)=0$ and every $t\in\mathbb{R},\,\mu(\exp(tf)) \leq \exp\left(\frac{ct^2}{2}\right).$
\end{theorem*}
\begin{example}
Let $\X$ be a discrete space and $d(x,y)=\mathds{1}_{x\neq y}$. We have that $W_1(\mu,\nu) = TV(\mu,\nu)$, \textit{i.e.}, the Total Variation distance between $\mu,\nu$ \cite[Prop. 3.4.1]{concentrationMeasureII}. In this case the transportation cost inequality is well-known under the name of Pinsker's inequality and it holds for every $\mu\in\Pm_1(\X)$ with $c=1/4$. In general, such inequalities are highly non-trivial and hold for specific distributions (e.g., Gaussian, etc.).
\end{example}
\iffalse
Moreover, transportation-cost inequalities with $p\in[1,2]$ tensorise.
\begin{proposition}\cite[Prop. 3.4.3] {concentrationMeasureII}\label{tensorization} Let $p\in[1,2]$.
If $\mu$ satisfies a $T_p(c)$ inequality on $(\X,d)$ then the product measure $\mu^{\otimes n}$ satisfies a $T_p(cn^{\frac{2}{p} -1})$ inequality on $(\X^n,d_{p,n})$ where $d_{p,n}(x^n,y^n) = \left(\sum_{i=1}^n d(x_i,y_i)^p\right)^{1/p}$.
\end{proposition}\fi
\subsection{Learning Theory} \iffalse
We are mainly interested in supervised learning, where the algorithm learns a \emph{classifier} by looking at points in a proper space and the corresponding labels.
Suppose we have an instance space $\mathcal{Z}$ and a hypothesis space $\mathcal{H}$. The hypothesis space is a set of functions that, given a data point $z\in \mathcal{Z}$ outputs the corresponding label $y_z \in \mathcal{Y}$. Suppose we are given a training data set $\mathcal{Z}^n\ni S =\{z_1,\ldots,z_n\}$ made of $n$ points sampled in an i.i.d. fashion from some distribution $\mathcal{P}$. Given some $n\in\mathbb{N}$, a learning algorithm is a (possibly stochastic) mapping $\mathcal{A}:\mathcal{Z}^n\to \mathcal{H}$ that given as an input a finite sequence of points $S\in\mathcal{Z}^n$ outputs some classifier $h=\mathcal{A}(S)\in\mathcal{H}$.
In the simplest setting we can think of $\mathcal{Z}$ as a product between the space of data points and the space of labels \textit{i.e.}, $\mathcal{Z}=\mathcal{X}\times\mathcal{Y}$ and suppose that $\mathcal{A}$ is fed with $n$ data-label pairs $(x,y)\in\mathcal{Z}$. In this work we will view $\mathcal{A}$ as a family of conditional distributions $\mathcal{P}_{H|S}$ and provide a stochastic analysis of its generalization capabilities.
The goal is to generate a hypothesis $h:\mathcal{X}\to \mathcal{Y}$ that has good performance on both the training set and newly sampled points from $\mathcal{X}$. To guarantee this, one has to keep the generalization error bounded.\fi 
For reasons of space we skip a thorough description of the learning framework. The setting we consider is the one described in \cite[Section II.C]{fullVersionGeneralization}. In particular, in this work, we focus on the expected generalisation error which can be defined as follows:
\begin{definition}\label{generr} Let $\mathcal{P}_Z$ be a probability measure over $\mathcal{Z}$. Let $\ell:\mathcal{H}\times\mathcal{Z}\to\mathbb{R}$ be a loss function. The error (or risk) of a prediction rule $h$ with respect to $\mathcal{P}$ is defined as $L _\mathcal{P}(h)=\Pm_Z(\ell(h,Z))$
	while, given a sample $S=(z_1,\ldots,z_n)$, 
	the empirical error
	of $h$ with respect to $S$ is defined as $L_{S}(h) = \frac1n \sum_{i=1}^n \ell(h, z_i).$
	Moreover, given a learning algorithm $\mathcal{A}:\mathcal{Z}^n\to\mathcal{H}$, its generalization error with respect to $S$ is defined as follows:
	\begin{equation}\label{eq:generr}\text{gen-err}_\mathcal{P}(\mathcal{A},S)=\Pm_{HS}(L_{S}(\mathcal{A}(S))-L_{\mathcal{P}}(\mathcal{A}(S))).\end{equation}
\end{definition}
\section{Main Result}\iffalse
It is possible to show that  given a strictly convex and differentiable function over $A$, $\phi$, then
\begin{equation}
    \phi^\star(x) = x{\left(\phi'\right)}^{-1}(x)-\phi(\left({\phi'}\right)^{-1}(x)),\,\,\forall x\in A\label{convexFunctEquality2}.
\end{equation}\fi

%Let us begin with a technical lemma first and then proceed right away  with the main result. We will then elaborate on it to show how it is connected to both generalisation error and transportation-cost inequalities.
Let us start with the main result and then elaborate on it to show how it is connected to both generalisation error and transportation-cost inequalities.
\iffalse
\begin{lemma}[{{\cite[Eq. (1.11)]{minimizationMeasures}}}]
Let $\phi$ be a strictly convex and differentiable function over $A$, then
\begin{equation}
    \phi^\star(x) = x{\left(\phi'\right)}^{-1}(x)-\phi(\left({\phi'}\right)^{-1}(x)),\,\,\forall x\in A\label{convexFunctEquality2}.
\end{equation}
\end{lemma}\fi
%In order to have a more thorough Theorem 1 it might be worth exploring to what extent one can use \eqref{convexFunctEquality2} with different notions of duality. With Young's duality, for instance, we do not need this result at all as we can solve the infimum in \eqref{toMinimiseExpBound} using \cite[Lemma 2.4]{BLM2013Concentration}.
\begin{theorem}\label{generalBoundInverseConjugate}
Let $f\in C_c(\X)$ and $\phi$ be a strongly convex function such that its Legendre-Fenchel dual $\phi^\star$ admits a generalised inverse ${\phi^{\star}}^{-1}(t)=\inf\{s: \phi(s)\geq t\}$. Let $\psi^\star:\mathcal{M}(\X)\to\mathbb{R}$ be the Legendre-Fenchel dual of $\psi:C_c(\X)\to\mathbb{R}$. If \begin{equation} \psi(\lambda f)\leq \phi(\lambda)\text{ for every }\lambda>0\label{eq:generalBoundInverseConjugateAssumption} \end{equation} then, for every measure $\nu$ such that $\nu(f)<+\infty$ and ${\phi'}^{-1}\left({\phi^\star}^{-1}(\psi^\star(\nu))\right)>0$,
\begin{equation}
    \nu(f) \leq {\phi^\star}^{-1}\left(\psi^\star(\nu)\right).
    \label{eq:generalBoundInverseConjugateThesis}
\end{equation}
\end{theorem}
\begin{proof}
By definition of Legendre-Fenchel transform we know that for a given function $f$ and any given measure $\nu$:
\begin{align}
    \psi^\star(\nu) &= \sup_{f \in C_c(\X)} \langle  f,\xi \rangle - \psi(f)  \label{LFdual1}
    \\&\geq \lambda \langle f, \nu \rangle - \psi(\lambda f) = \lambda \nu(f)-\psi(\lambda f).
\end{align}
%\begin{align}
%    \psi^\star(\lambda f) &= \sup_{\xi \in \M(\X)} \langle \lambda f,\xi \rangle - \psi(\xi)  \label{LFdual1}
%    \\&\geq \lambda \langle f, \nu \rangle - \psi(\nu) = \lambda \nu(f)-\psi(\nu).
%\end{align}
Hence, we can say that, given $f\in C_c(\X)$, $\nu\in \M(\X)$ and $\lambda>0$
\begin{align}
     \nu(f) &\leq \frac{\psi(\lambda f) + \psi^\star(\nu)}{\lambda} \leq \frac{\phi(\lambda) + \psi^\star(\nu)}{\lambda}\label{toMinimiseExpBound},
\end{align}
where in \eqref{toMinimiseExpBound} we used our assumption on $\psi$. 
Denoting with $c=\psi^\star(\nu)$ and choosing $\lambda = {\phi'}^{-1}\left({\phi^\star}^{-1}(c)\right)$ gives us that
\begin{equation}
    \nu(f)\leq {\phi^\star}^{-1}(c) = {\phi^\star}^{-1}(\psi^\star(\nu)). \label{eq:inverseConjTarget}
\end{equation}
Indeed, let us denote, for simplicity, ${\phi^\star}^{-1}(c)=t$, then replacing $\lambda$ with ${\phi'}^{-1}\left({\phi^\star}^{-1}(c)\right)$ in \eqref{toMinimiseExpBound} one has
\begin{align}
    \frac{c+\phi({\phi'}^{-1}\left(t\right))}{{\phi'}^{-1}\left(t\right)} &= \frac{c+t\phi'^{-1}(t)-\phi^\star(t)}{{\phi'}^{-1}\left(t\right)} \label{crazyStep} \\
    &= t +\frac{c-\phi^\star(t)}{{\phi'}^{-1}\left(t\right)} \\&= {\phi^\star}^{-1}(c)  +\frac{c-\phi^\star({\phi^\star}^{-1}(c))}{{\phi'}^{-1}\left(t\right)} \\ &\leq  {\phi^\star}^{-1}(\psi^\star(\nu)) \label{lastStepCrazyProof},
\end{align}
where \eqref{crazyStep} follows from \cite[Eq. (1.11)]{minimizationMeasures}.
\end{proof}

\begin{remark}
In most settings of interest the assumption that ${\phi'}^{-1}\left({\phi^\star}^{-1}(\psi^\star(\nu))\right)>0$ is easily satisfied. If one considers $\psi^\star(\nu)$ to be a $\varphi$-Divergence, \textit{i.e.,} $\psi^\star(\nu)=D_{\varphi}(\nu\|\mu)$, with $\mu$ a probability measure fixed before-hand,
then $\psi^\star(\nu)\geq 0$ for every $\nu\in\Pm(\X)$. Typical functions $\phi$ will be of the form $\phi(x)=x^{\alpha}/\alpha$ with $\alpha>1$. This implies that ${\phi'}^{-1}\left({\phi^\star}^{-1}(\psi^\star(\nu))\right)=(\beta \psi^\star(\nu))^{\frac1\alpha}$ with $\beta= \alpha/(\alpha-1)>0$, which is clearly positive for $\psi^\star(\nu)>0$.
\end{remark}\iffalse
\begin{remark}
The assumption that $\mu(f)=0$ is not technically necessary in any of the arguments we will use throughout the paper. It does, however, provide for ``cleaner statements''. \MG{Where are you making this assumption?}
\end{remark}\fi
%\begin{remark}
A variety of assumptions in Theorem \ref{generalBoundInverseConjugate} can be altered. It is possible, for instance, % to relax the assumption of convexity of $\psi$ and 
to choose different notions of duality regarding both $\psi$ and $\phi$. %changing the space over which we take the supremum in \eqref{LFdual1}. For instance,  
If instead of Legendre-Fenchal duality of $\phi$ one uses the notion of Young's duality/Young's complementary functions then it is actually possible to find the infimum over $\lambda>0$ in \eqref{toMinimiseExpBound} using \cite[Lemma 2.4]{BLM2013Concentration} (modulo some extra assumptions on $\phi$). %Moreover, these extra assumptions on $\phi$ automatically imply that ${\phi'}^{-1}\left({\phi^\star}^{-1}(\psi(\nu))\right)>0$. %Moreover, this does not require $\phi$ to be strongly convex, but convexity will suffice.
%Assuming $\phi$ to be a Young function then guarantees the positivity of its dual $\phi^\star$ and of its derivative and consequently of the respective inverses.
The result remains unchanged with a slightly different set of assumptions and the knowledge that such bound is the tightest with respect to $\lambda$ in \eqref{toMinimiseExpBound}.
Other choices might be suitable as well. 
For instance, selecting $\psi_{\mu}^\star(f)=\mu(\exp(f))-1$ one recovers $\psi^\star_{\mu}(\nu)$ to be the KL-Divergence between any positive measure $\nu$ (not necessarily a probability measure) and a fixed measure $\mu$. Alternatively, considering $\psi_\mu(\nu)=\log(\mu(\exp(f)))$ then $\psi^\star_\mu$ would be the usual KL-Divergence, defined only on probability measures, (e.g., with $D(\nu\|\mu)=+\infty$ if $\nu\not\in\Pm(\X))$. In general, whenever we restrict the domain of the divergence and,  we obtain a potentially tighter variational representation of $\psi^\star$ (in the considered space) which might, however, be harder to compute. In the case of KL the restriction to probability measures is well-known in the literature as the ``Donsker-Varadhan'' representation of KL. For more details on the matter we refer the reader to~\cite{varReprProbMeasures}. 
%\end{remark}

Let us now proceed with examples with the purpose of better understanding both the assumptions and the implications of Theorem \ref{generalBoundInverseConjugate}. %The most important objects are arguably $\psi$ and $\phi$. %The very first assumption requires $\psi$ to be a convex functional.  (Thinking about it, we might not need it, $\psi^\star$ can be defined nonetheless. Moreover, $\psi^\star$ does not need to be the Legendre-Fenchel dual as well. We can define it to be the $\sup$ over whatever space). 
%Let us consider $\psi_\mu(\cdot)$ to be Kullback-Leibler divergence defined over $\Pm(\X)$ and consequently consider %\MG{why ``let us consider''? shouldn't we write ``for which the dual is well known to be...''} \textcolor{blue}{AE: In Remark 2 I talk about the possibility of different ``duals'' according to the space one supremises over. If one supremises over probability measures, then the dual is the following one, otherwise, over general measures,it is $\mu(\exp(f-1))$.}\MG{Okay, if you want; but it is only inside the remark. Almost suggests that maybe Remark 2 should not be a remark, but part of the main discussion. (A ``remark,'' by convention, is something that can be skipped and the rest of the paper still makes sense.} 
%$\psi_{\mu}^\star(f)= \log\mu(\exp(f))$.
\begin{corollary}\label{boundDifferenceKL}
Let $X$ be a random variable over the probability space $(\X,\mathcal{F},\mu)$ and assume $X$ to be a zero-mean $\sigma^2$-sub-Gaussian random variable \footnote{Given a zero-mean random variable $X$ we say that it is  $\sigma^2$-sub-Gaussian if the following holds true for every $\lambda\in\mathbb{R}$ : $
\log(\mu(e^{\lambda X})) \leq \frac{\lambda^2\sigma^2}{2}
$.} with respect to $\mu$. We have that for every measure $\nu$ such that $\nu\ll\mu$ and such that $\nu(X)<+\infty$
\begin{equation}
    \nu(X) \leq \sqrt{2\sigma^2 D(\nu\|\mu)}. \label{boundDiffKL}
\end{equation}
\end{corollary}
\begin{proof} 
The proof follows from Theorem \ref{generalBoundInverseConjugate}, selecting $\psi(f)= \log(\mu(\exp(f)))$ and thus $\psi^\star(\nu)=D(\nu\|\mu)$ with $\nu\in\Pm(\X)$. %and selecting then $\phi(\lambda)=\frac{\lambda^2\sigma^2}{2}$,
%which is a strongly convex function.
The assumption that $X$ is sub-Gaussian under $\mu$ implies that $
    \psi(\lambda X) \leq \frac{\lambda^2\sigma^2}{2} = \phi(\lambda)$ for every $\lambda\in\mathbb{R}$, where $\phi$ is a strongly convex function.
To complete the argument, we observe that $
    \phi^\star(\lambda^\star)=\frac{{\lambda^\star}^2}{2\sigma^2} \implies   {\phi^\star}^{-1}(t) = \sqrt{2\sigma^2t},$
which establishes the claimed bound.
\end{proof}
%Generalisations of this type of inequalities beyond sub-Gaussianity are presented in Section ...
\iffalse
Corollary \ref{boundDifferenceKL} represents one of the simplest examples of applicability of Theorem \ref{generalBoundInverseConjugate}. The choice of $\psi_\mu$ and $\psi_\mu^\star$ allows us to immediately retrieve a well-known result on $\sigma^2$-sub-Gaussian random variables. This result lies at the foundations of most of the bounds connecting Mutual Information and the Expected Generalisation Error of a learning algorithm.\fi
\section{Generalisaton Error Bounds}
General bounds can be provided using Theorem \ref{generalBoundInverseConjugate} and the well-known Corollary \ref{boundDifferenceKL}. However, to offer something even more meaningful to the reader, let us set-up a very specific framework and restrict our assumptions further. Given a convex functional $\varphi$, we will consider the following family of functionals over measures $\psi^\star_\mu(\cdot)=D_{\varphi}(\cdot\|\mu)$, with $\mu$ a measure fixed beforehand. Consequently $\psi_\mu$ denotes the Legendre-Fenchel dual of $D_\varphi(\cdot\|\mu)$~\cite[Lemma 4.5.8]{dembo2009large}. We will assume, as above, that given a function $f\in C_c(\X)$, $\psi_\mu(\lambda f) \leq \phi(\lambda)$ for some convex function $
\phi$ and every $\lambda>0$.
\iffalse\begin{definition}
We will say that $\psi^\star$ and $\phi$ satisfy the $n$-sum property if,
given a sequence $(f_1,\ldots,f_n)\in{(\M(\X)^\star)}^n$ and assuming that $\psi^\star(\lambda f_i) \leq \phi(\lambda)$ for every $1\leq i \leq n$ then \begin{equation} \psi^\star\left(\lambda \frac1n \sum_i f_i\right) \leq \frac{\phi(\lambda)}{n}.
\end{equation}
\end{definition}\fi
 Let us set ourselves in a classical learning setting. Let $S=(Z_1,\ldots,Z_n)$ and $H=\A(S)$ be two random variables respectively over the spaces $(\Z^n,\mathcal{F}_{\Z^n}, \Pm_S\!\!=\!\!\Pm_Z^{\otimes n})$, $(\Hy,\mathcal{F}_\Hy, \Pm_H)$. Let $\Pm_S\Pm_H$ denote the joint measure induced by the product of the marginals of $S$ and $H$ over $(\Z^n\times\Hy, \mathcal{F})$. 
Let $S$ be the input of a learning algorithm $\A$ and $H=\A(S)$ the corresponding output, let $\Pm_{SH}$ be the joint measure induced by $\A$. Assume, moreover, that $\Pm_{SH}\ll \Pm_S\Pm_H$.

In order to match the framework that is typically utilised in the Learning Theory literature, we will work with $\psi^\star_{\Pm_S\Pm_H}(\cdot)=D_\varphi (\cdot\|\Pm_S\Pm_H)$ and, following the structure of Theorem \ref{generalBoundInverseConjugate}, one has to assume something about $\psi_{\Pm_S\Pm_H}$ in order to provide a bound that involves $\psi^\star_{\Pm_S\Pm_H}$. However, our assumptions (again, matching the literature) will involve $\psi_{\Pm_S}$ or $\psi_{\Pm_Z}$.  The reason why we can do this is the following.  Given the choice of $\psi^\star_\mu = D_\varphi(\cdot\|\mu)$, we typically know the shape that $\psi_\mu$ will have. If $\psi^\star_\mu$ is the KL divergence $\psi_\mu(f)=\log\mu(\exp(f))$ while if $\psi^\star_\mu$ is a $\varphi-$Divergence $\psi_\mu(f)=\mu(\varphi^\star(f))$ ~\cite{minimizationMeasures}. This naturally implies that if we consider product measures, one has the following characterisation for the dual $\psi_{\mu\times\xi}(f)=\xi(\mu(\varphi^\star(f)))=\xi(\psi_\mu(f))$ for $\varphi$-Divergences and $\exp(\psi_{\mu\times\xi}(f))=\xi(\mu(\exp(f)))=\xi(\exp(\psi_{\mu}(f)))$ for KL.

Consequently, since we will consider the product measure $\Pm_S\Pm_H$ we have that, given the structure of the dual, an upper-bound on $\psi_{\Pm_S}$ for every $h\in\mathcal{H}$ naturally implies an upper-bound on $\psi{\Pm_S\Pm_H}$. Another important consideration is that, an assumption of the form $\psi_{\Pm_Z}(\lambda(\ell(h,Z)-\Pm_Z(\ell(h,Z))))\leq \phi(\lambda)$ for every $h$ typically implies an assumption of the form $\psi_{\Pm_S}(\frac{\lambda}{n}(\sum_i (\ell(h,Z_i)-\Pm_S(\ell(h,Z_i)))))\leq \phi(\lambda)/n$, with $S$ a sequence of $n$ iid samples distributed according to $\Pm_Z^{\otimes n}$, something that we will informally define as the ``n-sum property'' of $\psi_{\Pm_Z},\psi_{\Pm_S}$ and $\phi$.
Under this framework, we can state the following result.
\iffalse \textcolor{blue}{A few considerations: if one wants the next statement to be general, it is necessary to make an assumption about $\psi^\star_{\Pm_S\Pm_H}$. The issue is that in the next corollary I am not doing this but I am implicitly using the structure of the dual of a $\varphi$-Divergence. One typically has that $\psi^\star_\mu(f) = \mu(\varphi^\star(f))$ and thus, $\psi^\star_{\mu\times \xi}(f)= \xi(\psi^\star_\mu(f)).$ This is why an assumption like $\psi^\star_{\Pm_Z}(\lambda(\ell(h,Z)-\Pm_Z(\ell(h,Z))))) \leq \phi(\lambda)$ for every $\lambda$ and $h$ provides guarantees in terms of $\psi^\star_{\Pm_S\Pm_H}(f)$ and allows us to use Theorem \ref{generalBoundInverseConjugate} in order to obtain the Corollary. Now, the question is, should I add a paragraph at the beginning of the section, right after having fixed $\psi=D_\varphi$ for some $\varphi$, where I make these points? Technically, with KL, the reasoning is yet another one because we have that $\psi^\star_\mu(f)=\log\mu(\exp(f))$ and $\psi^\star_{\mu\times\xi}(f) = \log \xi(\mu(\exp(f))))$ and it still works but thanks to the non-increasability of $\log$.}
\fi
\begin{corollary}\label{generalGenErrBound}
 Let $\psi^\star_{\Pm_S\Pm_H}(\cdot)=D_\varphi(\cdot\|\Pm_S\Pm_H)$ for some convex functional $\varphi$. Assume that $\psi_{\Pm_S\Pm_H}(\lambda(L_S(H)-\Pm_H(L_\Pm(H))))\leq \phi(\lambda)/n$ for every $h\in \mathcal{H} $ and $n,\lambda>0$. One has that
\begin{equation}
    \text{gen-err}_\mathcal{P}(\mathcal{A},S) \leq \frac{{\phi^\star}^{-1}(nD_\varphi(\Pm_{SH}\|\Pm_S\Pm_H))}{n}. 
\end{equation}
\end{corollary}
\begin{proof}
%As per assumption, $\psi^\star_{\Pm_S\Pm_H}(\lambda(L_S(H)-\Pm_H(L_\Pm(H))))\leq \phi(\lambda)/n$.
Let us denote with $\phi_n(\lambda)=\frac{\phi(\lambda)}{n}$ one has that $\phi_n^\star(\lambda^\star) = \frac1n\phi^\star(\lambda n)$ and consequently ${\phi_n^\star}^{-1}(t)= \frac1n{\phi^\star}^{-1}(nt).$ The statement then follows from Theorem \ref{generalBoundInverseConjugate}.
\end{proof}
From this we can now easily recover all the bounds on the expected generalisation-error that involve KL and that are present in the literature, for instance:%(or, given the generality of Theorem \ref{generalBoundInverseConjugate}, virtually any Divergence).% While, in order to bridge the gap to transportation-cost inequalities we need a bit more work.
\begin{corollary}\label{boundGeneralErrorMI}
Let $\varphi(x)=x\log x$ in Corollary \ref{generalGenErrBound}, hence $D_\varphi(\cdot\|\Pm_S\Pm_H)$ is the KL-Divergence.  Assume that the loss function $\ell(h,Z_i)$ %as defined in Definition \ref{generr}
is $\sigma^2$-sub-Gaussian under $\Pm_Z$ for every $h$. We have that
\begin{equation}
    \text{gen-err}_\mathcal{P}(\mathcal{A},S) \leq \sqrt{\frac{2\sigma^2I(S;\A(S))}{n}}. \label{eq:genErrMI}
\end{equation}
\end{corollary}
\begin{proof}
Since $\ell(h,Z_i)$ is $\sigma^2$-sub-Gaussian under $\Pm_Z$, then \iffalse$$\psi^\star_{\Pm_Z}(\lambda(\ell(h,Z_i)-\Pm_Z(\ell(h,Z)))) \leq \frac{\lambda^2\sigma^2}{2}$$ for every $i$, where $\psi^\star_{\Pm_Z}(f)=\log\Pm_Z(\exp(f))$.
 Moreover, one also has that\fi $\psi_{\Pm_S}\left(\frac1n\sum_i\lambda(\ell(h,Z_i)-\Pm_Z(\ell(h,Z)))\right)\leq \frac{\phi(\lambda)}{n}.$
 This implies that $\psi_{\Pm_Z},\psi_{\Pm_S}$ and $\phi$ satisfy the ``$n$-sum property''. The argument then follows directly from Corollary \ref{generalGenErrBound} along with the fact that $\frac{1}{n}{\phi^\star}^{-1}(nt)=%\frac1n \sqrt{2\sigma^2nt}=
 \sqrt{\frac{2\sigma^2t}{n}}.$ 
\end{proof}
\begin{remark}
Both Corollary \ref{boundDifferenceKL} and \ref{boundGeneralErrorMI} have appeared in the literature in a variety of forms \cite[Prop. 1]{explBiasMI}, \cite[Lemma 1]{infoThGenAn}. 
\end{remark}\iffalse
\begin{remark}
In a learning theory framework this essentially means that for a given divergence $D_\varphi$, if we can control $\psi_{\Pm_Z}^\star(\lambda \ell(h,Z))$ for every $h$ (or we can control $\psi^\star_{\Pm_H\Pm_Z}(\lambda\ell(H,Z)))$, on average with respect to $\Pm_H$)
using a convex function $\phi$, such that ${\phi^\star}^{-1}(x)$ is sub-linear in $x$, then we can aim at a generalisation error bound that decays with $n$. The effective decay, though, will then depend on the behaviour of $D_\varphi(\Pm_{SH}\|\Pm_S\Pm_H)$ as a function of $n$.
\end{remark}\fi
 To conclude the section and emphasise the generality of this approach, let us choose a different divergence. 
\begin{corollary}\label{genErrBoundChi}
Let $\varphi(x)=\frac{x^2}{2}$ in Corollary \ref{generalGenErrBound}, hence $\psi^\star_{\Pm_S\Pm_H}(\cdot)=D_\varphi(\cdot\|\Pm_S\Pm_H) = (\chi^2(\cdot\|\Pm_S\Pm_H)+1)/2$.
Assume that given the loss function $\ell(h,Z_i)$ there exists a constant $K>0$ such that for every $h$ $\Pm_Z((\lambda(\ell(h,Z_i)-\Pm_Z(\ell(h,Z)))^2)\leq K\lambda^2$ We have that
\begin{equation}
    \text{gen-err}_\mathcal{P}(\mathcal{A},S) \leq \sqrt{\frac{2K(\chi^2(\Pm_{SH}\|\Pm_S\Pm_H)+1)}{n}}.
\end{equation}
\end{corollary}
\begin{proof}
Given that $\psi^\star_{\Pm_S\Pm_H}(\nu) = \frac12 (\chi^2(\nu\|\Pm_S\Pm_H)+1)$ one has that $\psi_{\Pm_S\Pm_H}(f)=\frac12\Pm_S\Pm_H(f^2)$. Since the $Z_i$ are assumed to be iid random variables the variables $\ell(h,Z_i)$ are also iid and by assumption one has that  $\Pm_S((\lambda\frac{1}{n}\sum_i(\ell(h,Z_i)-\Pm_Z(\ell(h,Z)))^2)\leq \frac{\lambda^2}{n}K$, hence $\psi_{\Pm_Z},\psi_{\Pm_S}$ and $\phi$ satisfy the ``$n$-sum property''. The argument then follows from Corollary \ref{generalGenErrBound} and by noticing that $\phi^\star(\lambda^\star)=\frac{{\lambda^\star}^2}{2K}$ which in turn implies that ${\phi^\star}^{-1}(t) = \sqrt{2Kt}.$
\end{proof}
\begin{remark}\label{boundedMomentVSboundedCGF}
Assuming the $\sigma^2$-sub-Gaussianity of $\ell(h,Z_i)-\Pm_Z(\ell(h,Z))$ naturally implies a bound on $\Pm_Z((\lambda(\ell(h,Z_i)-\Pm_Z(\ell(h,Z)))^\kappa)$ for every $\kappa\geq 1$. For instance, $\sigma^2$-sub-Gaussianity implies that in Corollary \ref{genErrBoundChi},  $\Pm_Z((\lambda(\ell(h,Z_i)-\Pm_Z(\ell(h,Z)))^2)\leq K\lambda^2$ holds with $K=2\sigma^2 e^{2/e}$. However, assuming that $\ell$ has bounded variance (for every $h$ or on expectation wrt $\Pm_H$) is much less restrictive than assuming it is sub-Gaussian. As an example of this, suppose that the loss $\ell$ is $\sigma^2$-sub-Gaussian. One has that $\ell^2$ is sub-Exponential with parameters $(4\sqrt{2}\sigma^2,4\sigma^2)$. This means that $\ell^2$ has a finite log-moment generating function only for $\lambda<1/(4\sigma^2).$ Considering the framework given by Theorem \ref{generalBoundInverseConjugate} and Corollary \ref{boundGeneralErrorMI}, in order to solve the infimum in \eqref{toMinimiseExpBound} (which, in turn, provides Eq \eqref{eq:genErrMI}), one needs to select $\lambda^\star = \sqrt{\frac{2I(S;H)}{\sigma^2}}.$ Consequently, if $I(S;H)\geq \frac{1}{2^5\sigma^2}$ the bound is not valid as, for that choice of $\lambda$, the log-moment generating function of $\ell$ with respect to $\Pm_S\Pm_H$ is actually unbounded (and, thus, cannot be bounded by $\phi(\lambda)$). On the other hand, a sub-Exponential random variable is such that all of its moments are bounded,  hence, an approach as the one suggested by Corollary \ref{genErrBoundChi} can be successful.
\end{remark}
\subsection{Recovering other known results}
Other known results can be recovered, for instance noticing that $\text{gen-err}_\mathcal{P}(\mathcal{A},S)$ can be re-written as  $\frac1n \sum_{i=1}^n \left(\Pm_{HZ}(\ell(H,Z_i))-\Pm_H\Pm_Z(\ell(H,Z))\right)$ and using Corollary \ref{boundDifferenceKL} on each term inside the summation gives us the main result in \cite[Thm. 2]{ismi}.

Alternatively, the observation that one can consider a super-sample $\tilde{S}$ of length $2n$ and take two random subsets of length $n$ $S,S'$ independently, allows us to rewrite the generalisation error as follows $\text{gen-err}_\mathcal{P}(\mathcal{A},\tilde{S}) = \Pm_{\tilde{S}}(\Pm_{SH|\tilde{S}}(L_{\tilde{S}}(H)-L_\Pm(H)))$. One can then use Theorem \ref{generalBoundInverseConjugate} for every given choice of $\tilde{S}$ on $\Pm_{SH|\tilde{S}=\tilde{s}}(L_{\tilde{S}}(H)-L_\Pm(H))$ and recover \cite[Thm 5.1, Corollary 5.2, 5.3]{conditionalMI}. It represents a bound involving $D(\Pm_{SH|\tilde{S}=\tilde{s}}\|\Pm_{S|\tilde{S}=\tilde{s}}\Pm_{H|\tilde{S}=\tilde{s}})$ after assuming/showing that $\psi_{\Pm_{S|\tilde{S}=\tilde{s}}\Pm_{H|\tilde{S}=\tilde{s}}}(\lambda(L_{\tilde{s}}(H)-L_\Pm(H))) \leq \frac{\lambda^2}{2n}c(\tilde{s})$ with $c(\tilde{s})$ a constant that depends on the realisation of $\tilde{S}$. Clearly the possible combinations are numerous, however the pattern remains the same: a bound on the dual of the targeted divergence (or functional of measure) implies a bound on a difference of expectations. Indeed, still following the approach undertaken in \cite{ismi} but with the general spirit that characterises this work, one can show the following:
$
    \text{gen-err}_\mathcal{P}(\mathcal{A},S) \leq \frac1n \sum_{i=1}^n \sqrt{2K(\chi^2(\Pm_{Z_iH}\|\Pm_{Z_i}\Pm_H)+1)}.
$
Otherwise, following \cite{conditionalMI} one can recover bounds involving $\Pm_{\tilde{S}}\left(\chi^2\left(\Pm_{SH|\tilde{S}}\|\Pm_{S|\tilde{S}}\Pm_{H|\tilde{S}}\right)\right)$.
\section{Transportation-Cost inequalities}\label{phipIneq}
A recurring theme in the field that lies at the intersection between Information Theory and Optimal Transport is showing that Transportation-Cost Inequalities are equivalent to some form of concentration of measure. %For instance, if $p=1$ in Definition \ref{tpcIneq}, a very important result proven by Bobkov and Götze established that a $T_1(c)$ inequality for some measure $\mu$ is equivalent to concentration under $\mu$ of every $1$-Lipschitz function. 
One such example is Theorem$\,\diamond$ in Section \ref{tcI}.
An in-depth review of this connection can be found in \cite[Section 3.4]{concentrationMeasureII}.
This type of results are generally made of two parts: the ``if part'' (if $\mu$ satisfies concentration for every function in some family then the $T_p(c)$ holds) and the ``only if'' part (if $\mu$ satisfies a $T_p(c)$ inequality then one has concentration for every function in some family). So far we have essentially established the ``if part'' for general functionals of measures. Let us now establish the ``only if'' part and then discuss how it relates to the classical framework. %first and discuss it later. %We will restrict ourselves to Young functions for simplicity.
\begin{theorem}\label{onlyIfGeneralBoundInverseConjugate}
Let  %$\mu$ be a measure defined over some space $\X$ and assume that
$\psi:\M(\X)\to\mathbb{R}$ be a %\textcolor{blue}{most likely not needed:convex}
functional and 
let $f\in\M(\X)^\star$. If there exist a Young function $\phi$ and its complementary function  $\phi^\star$%(t)= \sup_{s\geq 0} s|t|-\phi(t)$, and $\phi^\star$ 
 which admits a generalised inverse ${\phi^\star}^{-1}$ such that for every $\nu$ with $\nu(f)<+\infty$ one has
\begin{equation}
    \nu(f) \leq {\phi^\star}^{-1}\left(\psi(\nu)\right), 
    \label{eq:onlyIfGeneralBoundInverseConjugateAssumpt}
\end{equation}
then 
\begin{equation}
    \psi^\star(\lambda f)\leq \phi(\lambda)\text{ for every }\lambda>0, \label{eq:onlyIfGeneralBoundInverseConjugateThesis}
\end{equation}
where $\psi^\star:\mathcal{M}(\X)^\star \to\mathbb{R}$ is the Legendre-Fenchel dual of $\psi$.
\end{theorem}
\begin{proof}
Given that $\phi$ is a Young function, one has that by \cite[Lemma 2.4]{BLM2013Concentration}
\begin{align}
    \nu(f) &\leq {\phi^\star}^{-1}\left(\psi(\nu)\right) =\inf_{\lambda>0} \frac{\phi(\lambda) + \psi(\nu)}{\lambda}.
\end{align}
%This means that for every $\lambda>0$,
%\begin{equation}
%    \nu(f) \leq  \frac{\phi(\lambda) + \psi(\nu)}{\lambda}.
%\end{equation}
%Or, that 
This means that for every $\lambda>0$ and every $\nu$-integrable function $f$,
\begin{equation}
    \phi(\lambda) \geq \lambda \nu(f) - \psi(\nu). \label{step1}
\end{equation}
Taking the supremum with respect to $\nu\in\M(\X)$ one recovers the statement:
\iffalse
From the definition of Legendre-Fenchel dual one has that for every function $g$
\begin{equation}
    \psi^\star(g) = \sup_{\nu\in\M(\X)} \nu(g)-\psi(\nu).
\end{equation}
In particular, setting $g=\lambda f$ one can show that the supremum is attained by $\hat{\nu}={\psi^\star}^{'}(\lambda f)$. (Absolute continuity of $\hat{\nu}$ with respect to $\mu$ still needs to be shown. By linearity of the Lebesgue Integral, one recovers %(ARGUE EXISTENCE OF THIS DERIVATIVE IN GENERAL SPACES AND ABSOLUTE CONTINUITY OF $\hat{\nu}$ WRT $\mu$).
\begin{equation}
    \psi^\star(\lambda f) = \lambda \hat{\nu}(f)-\psi(\hat{\nu}) \label{equalityLFDual}.
\end{equation}
Setting then $\nu=\hat{\nu}$ and plugging \eqref{equalityLFDual} in \eqref{step1} one recovers the statement: \fi
$
 \phi(\lambda) \geq \psi^\star(\lambda f). 
$
\end{proof}
In order to see how Theorem \ref{generalBoundInverseConjugate} and Theorem \ref{onlyIfGeneralBoundInverseConjugate} represent, respectively, the ``if'' and ``only if'' part connecting $T_p(c)-$like inequalities to concentration let us recover Theorem$\,\diamond$. 
In particular, select
\begin{itemize} 
\item $\psi^\star(\nu) = D(\nu\|\mu)$ for a given $\mu$ (consequently, one has that $\psi(\lambda f) =\log\mu(\exp(\lambda f))$, 
cf.~\cite[Lemma 6.2.13]{dembo2009large});
\item $\phi(\lambda)= \frac{c\lambda^2}{2}$ (which implies ${\phi^\star}^{-1}(\kappa)=\sqrt{2c\kappa}$);
\end{itemize}
Theorem~\ref{generalBoundInverseConjugate} is what allows us to reach \eqref{tpcThesis} starting from \eqref{tpcHypothesis}.
Keeping the same $\phi$ but inverting the roles of $\psi$ and $\psi^\star$ in Theorem~\ref{onlyIfGeneralBoundInverseConjugate} is what allows us to reach \eqref{tpcHypothesis} starting from \eqref{tpcThesis}. 
 Some extra technical steps are necessary in order to bring in Wasserstein Distances (which will be considered in the proof just below). Given the generality of the results we can, as an example, consider a setting similar to Theorem$\,\diamond$ but involving a different divergence:
\begin{theorem}\label{TpcHellinger}
Let $\mu \in \Pm_1(\X)$ and $\beta>1$. There exists a $c$ such that for every $\lambda$ and every $1$-Lipschitz function $f$ \begin{equation}\mu\left(|\lambda f|^\beta\right)\leq (c\lambda )^{\beta}\label{eq:TpcHellingerAssumption},\end{equation} if and only if, 
for every $\nu\ll\mu$%\iffalse\textit{i.e.}, $f$ is $\theta$-sub-Weibull\footnote{A random variable $X$ is said to be $\theta$ sub-Weibull with respect to a measure $\mu$ if there exists a $K$ such that 
%$
%    \log\mu\left(\exp{(\lambda X)^\frac1\theta}\right) \leq (\lambda K)^{\frac1\theta}={\phi^\star}^{-1}(H_\alpha(\nu\|\mu)).
%$
%If $\theta=\frac{1}{2}$ the random variable $X$ is said to be sub-Gaussian.
%If $\theta=1$ the random variable $X$ is said to be sub-Exponential.} with respect to $\mu$.  Then it is possible to show that
%\begin{equation}
%    W_1(\mu,\nu)\leq \left(\frac{D(\nu\|\mu)K^{\frac{1}{1-\theta}}}{\theta^{\frac{\theta}{1-\theta}}-\theta^{\frac{1}{1-\theta}}}\right)^{1-\theta}. \label{boundDiffKLsubWeib}
%\end{equation}

%then,
%\fi
\begin{align}
    W_1(\mu,\nu)&\leq \left(\alpha c^\alpha H_\alpha(\nu\|\mu)\right)^\frac{1}{\alpha}\label{boundDiffHellinger}, %\\
   % &= |K|\left\lVert \frac{d\nu}{d\mu}\right\rVert_{L^\alpha(d\mu)}\label{boundDiffHellinger},
\end{align}
where %$H_\alpha(\nu\|\mu)$ denotes $D_\varphi(\nu\|\mu)$ with $\varphi(x)= \frac{x^\alpha}{\alpha}$ and 
$\alpha= \frac{\beta}{\beta-1}.$
Setting $\beta=2$ we recover the following:
\begin{equation}
    W_1(\mu,\nu)\leq \sqrt{c^22(\chi^2(\nu\|\mu)+1)} \label{TpcChiSq}.
\end{equation}
\end{theorem}
\begin{proof}
Let $\mu$ be a probability measure and let $\psi^\star(\nu)=H_\alpha(\nu\|\mu)$ with $\alpha>1$, one has that $\psi(f)=\frac{\mu(|f|^\beta)}{\beta}$ (cf.~\cite[Theorem 3.3]{minimizationMeasures}). 
Moreover, let $\phi(\lambda) = \frac{|c\lambda|^\beta}{\beta}$, one has that $\phi_Y^\star(\lambda^\star)=\frac{|\lambda^\star/c|^\alpha}{\alpha}$ with $\alpha=\frac{\beta}{\beta-1}$. Consequently, for positive $\kappa$, ${\phi^\star}^{-1}(\kappa)= (\alpha c^\alpha \kappa )^\frac1\alpha.$ \newline Let $f$ be a function such that $\left\lVert f\right\rVert_{Lip}\leq 1$ and $\mu(f)=0$. Assume that for every $\nu\ll\mu$, $W_1(\mu,\nu)\leq(\alpha c^\alpha  $ $H_\alpha(\nu\|\mu))^\frac{1}{\alpha}$. By the Kantorovich-Rubenstein dual representation of $W_1$ one has that
\begin{equation}
    W_1(\mu,\nu) = \sup_{f:\left\lVert f\right\rVert_{Lip}\leq 1} \left|\mu(f)-\nu(f)\right|\label{KantorovichRubenstein}.
\end{equation}
Hence, for every function $f$ such that $\left\lVert f\right\rVert_{Lip}\leq 1$, $\mu(f)=0$ one can rewrite Equation~\ref{TpcChiSq} as follows:
\begin{equation}
    \nu(f) \leq (\alpha c^\alpha H_\alpha(\nu\|\mu))^\frac{1}{\alpha} = {\phi^\star}^{-1}(H_\alpha(\nu\|\mu)). \label{hypothesis}
\end{equation}
Consequently, by Theorem~\ref{onlyIfGeneralBoundInverseConjugate} one has that for every such $f$
\begin{equation}
    \psi^\star(\lambda f) = \mu\left(\frac{|\lambda f|^\beta}{\beta}\right) \leq \frac{(c\lambda )^\beta}{\beta} = \phi(\lambda) . \label{thesis}
\end{equation}
Similarly, assuming \eqref{thesis} leads to \eqref{hypothesis} for every such $f$ via Theorem~\ref{generalBoundInverseConjugate}. Repeating the same argument with -$f$ one reaches the following statement:
\begin{equation}
    |\mu(f)-\nu(f)|\leq (\alpha c^\alpha H_\alpha(\nu\|\mu))^\frac{1}{\alpha}. \label{toWasserstein}
\end{equation}
The assumption that $\mu(f)=0$, can be dropped replacing $f$ with $f-\mu(f)$. Taking then the supremum over all the $1-$Lipschitz functions $f$ in \eqref{toWasserstein} and using again the Kantorovich-Rubenstein duality formula for Wasserstein Distances one reaches that for every $\nu\ll\mu$, $$W_1(\mu,\nu) \leq (\alpha c^\alpha H_\alpha(\nu\|\mu))^\frac{1}{\alpha}.$$
\end{proof}
Drawing inspiration from Theorem~\ref{TpcHellinger} one can consider almost any $\varphi-$Divergence. Some restrictions on the possible choices of $\varphi$ arise naturally in order to have access to both variational representations~(cf.~\cite[Theorems 3.3, 3.4]{minimizationMeasures}).
\iffalse
\begin{remark}
$T_p(c)$-inequalities typically involve the KL-Divergence while in Theorem \ref{TpcHellinger} we focused on a family of divergences known as the Hellinger integrals. One can typically consider $\psi$ to be any $\varphi-$Divergence, \textit{i.e.},  $\psi_\mu(\nu)=D_\varphi(\nu\|\mu)$ for some convex functional $\varphi$. In order for such a Divergence to be properly defined, this requires us to impose the extra condition that $\nu\ll\mu$. In this case, in order to complete the proof of the ``only if'' part of Theorem \ref{TpcHellinger}, one also needs to argue that the supremum-achieving measure $\hat{\nu}$ is such that $\hat{\nu}\ll\mu$. This is clear from the proof given by Bobkov and Götze \cite{bobkov} when $\varphi(x)=x\log x$. In particular, for KL one has that given $\lambda$ and $f$, $d\hat{\nu}= \frac{\exp{(\lambda f)}d\mu}{\int\exp{(\lambda f)}d\mu}$. The same happens whenever one considers $\psi_\mu(\nu)=D_\varphi(\nu\|\mu)$ as, for a given $\lambda$ and $f$ and under a reasonable set of assumptions over $\varphi$ (\textit{e.g.}, strong convexity), %and according to the notion of duality that one uses,
one typically has that  $\hat{\nu}= {\varphi'}^{-1}(\lambda f)d\mu$.
\end{remark}
\fi
\begin{remark}
In the classical $T_1(c)$-setting one assumes that $\log(\mu(\exp(\lambda f)))\leq \frac{\lambda^2c^2}{2}$ for every $f$ that is $1-$Lipschitz and consequently recovers \eqref{tpcThesis}. % that $W_1(\mu,\nu)\leq \sqrt{2cD(\nu\|\mu)}={\phi^\star}^{-1}(D(\nu\|\mu))$ for every $\nu\ll\mu$. 
Considering Theorem \ref{TpcHellinger} instead, in Eq. \eqref{eq:TpcHellingerAssumption} we are only asking for the $\beta$-th moment to be bounded. The same argument as in Remark \ref{boundedMomentVSboundedCGF} holds: even though it is well known that $\chi^2(\nu\|\mu)\geq D(\nu\|\mu)$ (leading thus, to a worse bound on $W_1$), in order to get \eqref{TpcChiSq} we only need to bound the second moment of $f$ with respect to $\mu$ and such a bound can exist for functions with unbounded log-moment generating function, which are excluded from a classical $T_1(c)$ setting (cf. Remark \ref{boundedMomentVSboundedCGF}). \end{remark} 
Theorem \ref{TpcHellinger} highlights the following approach: 
starting from the variational representation of Wasserestein distances $W_p$ one understands the restriction on the family of functions that needs to be considered (cf. Equation \eqref{KantorovichRubenstein}). Let us denote such a family with $\mathcal{F}_p$. The next step is then to fix a measure $\mu$ and a functional $\psi^\star_\mu$ (in Theorem \ref{TpcHellinger}, the choice of $\psi^\star_\mu$ has fallen on the Hellinger integral). The
upper-bound that one can provide on $\psi(\lambda f)$ for $f\in\mathcal{F}_p$, characterised by $\phi(\lambda)$ (cf. Equations~\eqref{tpcHypothesis},~\eqref{eq:generalBoundInverseConjugateAssumption} and~\eqref{eq:TpcHellingerAssumption}), will determine the shape of the $T_p(c)$-like inequality (through ${\phi^\star}^{-1}$, cf. Equations  \eqref{tpcThesis}, \eqref{eq:generalBoundInverseConjugateThesis} and \eqref{boundDiffHellinger}), here denoted $\phi_p$-inequalities for convenience. Vice versa, bounding a Wasserstein distance $W_p$ through a divergence $\psi^\star_\mu(\nu)$ via a  $\phi_p$-inequality (cf. Equations  \eqref{tpcThesis}, \eqref{eq:onlyIfGeneralBoundInverseConjugateAssumpt} and \eqref{boundDiffHellinger}) implies a bound on the dual $\psi_\mu$ 
(cf. Equations \eqref{tpcHypothesis},  \eqref{eq:onlyIfGeneralBoundInverseConjugateThesis} and \eqref{eq:TpcHellingerAssumption}) and that can imply concentration according to $\psi_\mu$, $\phi$ and $\mathcal{F}_p$. %\\
\section{Conclusions} In this work we linked transportation-cost inequalities with generalisation error bounds. The thread connecting the two approaches is Legendre-Fenchel duality. As a result, we managed to generalise both approaches to various divergences and to random variables that are not necessarily sub-Gaussian. %The KL-Divergence can be replaced by any convex functional of measures and the resulting bound depends strongly on its Legendre-Fenchel dual. 
The approach undertaken in this work grants us two extra degrees of freedom: $\psi^\star$ and $\phi$. The functional $\psi^\star$ is not necessarily tied to be the KL-Divergence and the function $\phi$ can be any convex function that allows us to upper-bound $\psi(\lambda f)$. In the case of KL and its dual this is tantamount to showing concentration properties of $\mu$. Assuming a Gaussian-like behaviour of the log-moment generating function (hence, choosing $\phi(\lambda)= c\lambda^2/2$) leads to the familiar $\sqrt{2cD(\nu\|\mu)}$.
On the other hand, trying to show a classical $T_p(c)$-like inequality for some measure $\mu$ while fixing the shape of the inequality to be approximately $W_p(\mu,\nu)\leq  \sqrt{cD(\nu\|\mu)}$ (which means, essentially, fixing $\phi$) can be impossible. This depends on the concentration properties of $\mu$ or the family of functions that we have to consider $\mathcal{F}_p$. One can thus relax the inequality by either assuming a behaviour of the cumulant generating function that is different from Gaussian-like (same $\psi^\star$ but different $\phi$) or by picking a completely different divergence (changing $\psi^\star$ and/or $\phi$). % that, in turn, will require us to control different functionals of $f$, like in Theorem \ref{TpcHellinger}.
\iffalse
A very brief summary is the following: if one can bound $\psi_\mu^\star(\lambda f)$ for every $f\in\mathcal{F}_p$ then one is in the Transportation-Cost Inequality framework, if, instead, the bound on $\psi_\mu^\star$ is restricted to a specific function (or random variable), like in Corollary \ref{boundDifferenceKL}, \ref{generalGenErrBound}, \ref{boundGeneralErrorMI} and \ref{genErrBoundChi}, the setting is closer to the one exploited in Learning Theory.\fi 
\section*{Acknowledgment} The authors would like to thank Professor Ugo Vaccaro for insightful comments and suggestions on an early draft of this work. The work in this paper was supported in part by the Swiss National Science Foundation under Grant 200364.
% Keeping $\psi_\mu$ equal to the KL divergence but varying $\phi$ allows us to go beyond $\sigma^2$-sub-Gaussianity. Changing $\psi$ allows us to select different divergences. 
\bibliographystyle{IEEEtran}
\bibliography{sample}
\end{document}